\documentclass[runningheads]{llncs}
\pdfoutput=1
\usepackage{graphicx}
\graphicspath{{./figures/}{./diagrams/}}

\usepackage{booktabs} 
\usepackage[normalem]{ulem}
\usepackage{amsmath,amssymb,amsfonts}
\usepackage{enumitem}
\usepackage[super]{nth}
\usepackage{times}
\usepackage{textcomp}
\usepackage{multicol}

\usepackage[sort,nocompress]{cite}
\usepackage{soul}
\usepackage{color}

\usepackage{subfig}
\usepackage[export]{adjustbox}
\usepackage[hidelinks]{hyperref}

\ifdefined\pdfsuppresswarningpagegroup
\fi

\usepackage{mathtools}
\newtheorem{thm}{Theorem}
\newtheorem{lem}[thm]{Lemma}

\usepackage[ruled,linesnumbered]{algorithm2e} 
\SetAlFnt{\small}
\SetAlCapFnt{\small}
\SetAlCapNameFnt{\small}
\SetAlCapHSkip{0pt}
\IncMargin{-\parindent}


\hypersetup{pdfauthor={Liran Funaro, Orna Agmon Ben-Yehuda, and Assaf Schuster},pdftitle={Efficient Multi-Resource, Multi-Unit VCG Auction}}

\begin{document}

\setlength{\abovedisplayskip}{2pt}
\setlength{\belowdisplayskip}{2pt}

\title{Efficient Multi-Resource, Multi-Unit VCG Auction}
\author{Liran Funaro\inst{1} \and
Orna Agmon Ben-Yehuda\inst{1,2} \and
Assaf Schuster\inst{1}}
\authorrunning{L. Funaro et al.}
\institute{Computer Science Dept., Technion---Israel Institute of Technology\\
\email{\{funaro,ladypine,assaf\}@cs.technion.ac.il} \and
Caesarea Rothschild Institute for Interdisciplinary Applications of Computer Science, University of Haifa}

\maketitle 

\begin{abstract}

We consider the optimization problem of a multi-resource, multi-unit VCG auction that produces an optimal, i.e., non-approximated, social welfare.
We present an algorithm that solves this optimization problem with pseudo-polynomial complexity and demonstrate its efficiency via our implementation.
Our implementation is efficient enough to be deployed in real systems to allocate computing resources in fine time-granularity.
Our algorithm has a pseudo-near-linear time complexity on average (over all possible realistic inputs) with respect to the number of clients and the number of possible unit allocations.
In the worst case, it is quadratic with respect to the number of possible allocations.
Our experiments validate our analysis and show near-linear complexity.
This is in contrast to the unbounded, nonpolynomial complexity of known solutions, which do not scale well for a large number of agents.

For a single resource and concave valuations, our algorithm reproduces the results of
a well-known algorithm.
It does so, however, without subjecting the valuations to any restrictions and supports a multiple resource auction, which improves the social welfare over a combination of single-resource auctions by a factor of 2.5-50.
This makes our algorithm applicable to real clients in a real system.

\end{abstract}

\section{Introduction}

Infrastructure-as-a-Service (IaaS) providers have been using auctions to
control congestion via preemptible virtual-machine (VM) instances for nearly a decade~\cite{agmonben-yehuda2014ginseng,AmazonSpotInstances,alibabaAlibabaSpot,packetPacketSpot}.
A natural extension of this idea is to auction additional individual resources in an existing VM.
VCG auctions~\cite{clarke1971multipart,groves1973incentives,vickrey1961counterspeculation}
are appealing for this purpose, as they are \emph{truthful}: they incentivize clients to reveal their true valuation of the resources, which helps cloud providers accurately price their services.
Moreover, VCG maximizes the \emph{social welfare}---the aggregate valuation the clients assign to the chosen resource allocation.
For private (corporate) cloud providers, maximizing the social welfare maximizes the aggregate value the in-house clients generate for the corporation.
Cloud clients compete for multiple resources (e.g., RAM, CPU, bandwidth), and these need to be combined in a single auction.
A single resource VCG auction is computationally hard to solve~\cite{maille2007vcg}, and a multi-resource auction is more difficult.

Other solutions, besides auctions, were proposed for mitigating congestion. Posted prices~\cite{CloudSigmaPriceChart} and burstable performance~\cite{AmazonBurstableInstances,AzureBurstableVM,GoogleComputePricing,RackspaceCloudFlavors,CloudSigmaPricing}
incentivize clients to reduce their requirements and hence reduce the congestion.
Spot instances are based on the uniform price auction~\cite{agmonben-yehuda2013deconstructing}.
VCG (or generally affine-maximizer) mechanisms, however, are the only known truthful mechanisms that maximize social welfare~\cite{roberts1979characterization,lavi2003towards}.

The optimization problem for a single-resource VCG auction
can be reduced to a multiple-choice knapsack problem (MCK), which is NP-hard but can be solved in pseudo-polynomial time via dynamic programing~\cite{kellerer2004introduction}.
Many approximated, sub-optimal solutions have been proposed for the MCK problem~\cite{lawler1979fast,chekuri2005polynomial}.
However, for VCG to be truthful, an exact, optimal social welfare must be found~\cite{nisan2007computationally}.
To obtain a more efficient, exact solution for a single resource VCG auction, researchers relax the problem
by requiring all the functions that describe client valuations of a resource allocation (henceforth \emph{valuation functions}) to be monotonically increasing and concave~\cite{lazar1999design,maille2004multi} or usually concave~\cite{agmonben-yehuda2014ginseng}.
Others solve the problem for a single resource when only one function is not concave but is monotonically increasing~\cite{bae2008efficient}.
Concave valuation functions are an unrealistic requirement for cloud clients as their valuation functions have multiple inflection points~\cite{funaro2019stochastic,ye2015rochester,cameron2014we,wilkes2009utility,zhu2006utility,lee2007precise}.

To auction multiple resources, we must consider the relationship between them.
Usually, computing resources are complementary goods:
a client who is willing to pay one dollar for an additional single unit of CPU time and RAM is unwilling to pay anything for each resource individually.
Alternatively, the resources might be substitute goods:
a client who is willing to pay one dollar for an additional single unit of each resource is unwilling to pay two dollars for both resources together.
Thus, in both cases, the client cannot bid in an individual auction for each resource.
If this client partitions its budget between two resources, it may win only one or both.
A client pays for a worthless bundle if it wins only one of two complementary resources, or if it wins both substitute resources.
Such a scenario will also decrease the utilization.
Only a multiple resource auction that considers the clients' value for each combination of resources can both optimize the social welfare and be truthful.

Unfortunately, single resource solutions do not apply for multiple resources.
The multiple resource VCG auction can be reduced to a 
multiple-choice, \emph{multidimensional} knapsack problem (MCMK or d-MCK), which to the best of our knowledge has no pseudo-polynomial solutions.
Similarly to MCK, MCMK also has many approximated solutions~\cite{gao2017iterative,akbar2006solving,moser1997algorithm,khan2002solving,hifi2004heuristic}.
Such solutions provide near-optimal results: the best of them yields results within 6\% of the optimal value, which does not guarantee the auction will be truthful and maximize the social welfare.
Exact solutions for MCMK have been proposed via branch-and-bound algorithms (B\&B)~\cite{ghassemi2018exact,hifi2004exact,sbihi2007best,razzazi2008exact,gonen2000optimal}; 
however, their results indicate an implicit nonpolynomial increase in runtime with respect to the number of possible allocations.
These solutions were only tested empirically with small datasets and did not scale well for many clients and large, complete valuation functions.

Moreover, MCMK solutions were not designed for a VCG auction and thus do not allow efficient calculation of payments according to the VCG payment rule.
To compute a winning client's payment in a VCG auction, the auctioneer must find the social welfare that could be achieved when that winning client is excluded from the auction.
Solutions not tailored to VCG
must compute the payments by repeatedly finding the optimal allocation for each winning client
if that client had not participated in the auction.
This implies a worst-case quadratic complexity with respect to the number of clients.

In this work, we implement an efficient, exact, multi-unit, multidimensional resource VCG auction.
Two approaches can be considered for this problem.
The resources may be treated as infinitely divisible (continuous), as Lazar and Semret~\cite{lazar1999design}, Maill\'e and Tuffin~\cite{maille2004multi}, and Agmon Ben-Yehuda et al.~\cite{agmonben-yehuda2014ginseng} do for a single resource.
The other approach, which we adopt, divides each resource into
 identical units of a predefined size (e.g., a single CPU second can be time-shared as 1000 millisecond units).
The smaller the units are, the closer the auction's result is to the continuous solution, and the higher the complexity of finding the allocation that maximizes the social welfare.
  
In the multi-unit, multi-resource auction, agents, representing the clients, can bid using a multidimensional valuation function, which attaches a monetary value to each number of units of each resource.
To find the exact solution, the auctioneer must consider all the allocations for the number of agents and the number of resource units available.
Since the number of possible divisions of resources between agents is exponential in the number of agents and resource units, iterating over them is impractical. 

We present a method for solving a multi-unit, multi-resource auction 
without any restrictions on the valuation functions, 
in pseudo-near-linear time on average, over all possible realistic valuation functions, 
with respect to the number of clients ($n$) and the number of possible unit allocations for each client ($N$).
Our algorithm's worst-case time complexity is $O(n \cdot N^2)$, as opposed to the worst-case nonpolynomial complexity of the known MCMK algorithms.
Furthermore, our algorithm computes the VCG auction payments without repeating the full auction for each winning client.
The payment calculation complexity is a function of $N$ and the number of winning clients.
It does not depend on the number of clients in the auction ($n$).
Our solution is also applicable to a single resource auction and has a better average complexity than the dynamic programming solution, which is $O(n \cdot N^2)$~\cite{kellerer2004introduction}.
All of the above makes it feasible to choose a VCG auction as a resource allocation mechanism in a real system.

{\bf Our contributions} are an \emph {optimization algorithm} for the multi-unit, multi-resource allocation problem
and an implementation of this algorithm with a choice of data structures to support it.
We prove the correctness of the algorithm in Section~\ref{sec:joint-val-proof} and
numerically analyze its complexity in Section~\ref{sec:complexity}.
We evaluate the performance of our implementation in Section~\ref{sec:results} using each data structure and verify the correctness of the results.
We validate our results for a single resource with concave valuation functions, by comparing to Maill\'e and Tuffin's results, and show that separate single-resource auctions produce sub-optimal results, in contrast to multi-resource auctions, which produce optimal results.
The implementation can be extended using other data structures.
We analyze the algorithm's best possible performance independently of the choice of a data structure.

\section{The Non-Linear Optimization Problem}
\label{sec:nlop}
In this paper, vectorized arithmetic operators are defined element-wise.
For example,
$\vec a + \vec b = (a_0 + b_0, ..., a_n + b_n)$, and 
$\vec a \leq \vec b \Longleftrightarrow \forall i\in{1..R}: a_i\leq b_i$.
The symbols used in this paper are listed in Table~\ref{table:symbol}.
\begin{table}[b]
\centering
\caption{Table of symbols}
\label{table:symbol}
\begin{tabular}{ l l }
n & number of agents \\ \hline
R & number of resources \\ \hline
$\vec{m}$ & number of units for each resource: $(m_1, ... , m_R)$ \\ \hline
$\vec a_i$ & allocation of agent $i$ for each resource: $(a_{i,1}, ..., a_{i,R})$ \\ \hline
$A$ & set of allocations $\{\vec a_i\}_{i=1}^n$ \\ \hline
$V_i$ & valuation function of agent $i \in 1..n$ \\ \hline
$N$ & the number of possible allocations on which a valuation function is defined \\
    & $N = \prod_{r=1}^R{(m_r + 1)}$. \\
\end{tabular}
\end{table}

In an ideal VCG auction, the auctioneer computes the exact allocation that maximizes the social welfare.
Each winning client pays the auctioneer according to the damage it caused the rest of the clients---i.e., the \emph{exclusion compensation principle}.
This payment rule makes the auction \emph{truthful}: the best client strategy is to bid with its true valuation of the resources.
Thus, VCG optimizes the social welfare according to true data about client valuations.

The VCG optimization problem can be described as a non-linear optimization problem (NLP) that is
\emph{separable}, \emph{non-convex}, and \emph{linearly and discretely constrained}, as follows:

\noindent\textbf{Separable}:
The sum of $n$ separable valuation functions is maximized.
\begin{equation}
\text{Maximize:} \sum_i^n{V_i(\vec a_i)} \text{.}
\end{equation}
Such valuation functions can be represented as a multidimensional vector.

\noindent\textbf{Non-Convex}:
None of the separable functions ($V_i$) are required to be convex, concave, or even monotonic.

\noindent\textbf{Linearly Constrained}:
\begin{equation}
\sum_{i=1}^n{\vec a_i} \leq \vec m \text{.}
\end{equation}

\noindent\textbf{Discretely Constrained}: 
The resource is not continuous and is divided into units.
Each $a_{i,r}$ is a natural number (or zero) that represents the number of allocated units.
Only a whole unit can be allocated.
Hence, the $V_i$ functions should be defined only on an even-spaced grid of the natural numbers.

\section{Joint Valuation Algorithm}
\label{sec:joint-valuation-algorithm}
Funaro~et~al.~\cite{funaro2016ginseng} developed the \emph{joint valuation algorithm} for finding the optimal allocation of resources in a single dimension, for monotonically increasing functions with $O(n \cdot N^2)$ time complexity.
In this work, we extend this algorithm to multidimensional non-monotonic valuation functions, such that it fulfills all the constraints delineated in Section~\ref{sec:nlop}.
While the complexity of a na\"ive extension is proportional to the square of the number of possible unit-allocation combinations, our extension has a pseudo-near-linear complexity on average over all possible realistic valuation functions.

We%
 prove that the algorithm produces the correct optimal allocation and the correct payments in Section~\ref{sec:joint-val-proof}, and 
numerically analyze its time complexity in Section~\ref{sec:complexity}.

\subsection{Finding the Optimal Allocation}
To find the optimal allocation, two agents are first combined into one effective agent with a joint valuation function~(Section~\ref{sec:joining-vals}).
For any number and combination of goods that the two agents will obtain together, the joint function stores the optimal division of goods between them, and the sum of the valuations of these agents for this optimal division.
Then another agent is joined to the effective agent, and then another, etc.
This process produces a new joint valuation function at each stage, until the final effective agent's valuation function is the maximal aggregated valuation of all the agents.
Its maximal value is the maximal social welfare.
The optimal allocation is then reconstructed from the stored division data of the joint valuation functions.

\subsection{Payment Computation}
Our algorithm is efficient in the number of times that the optimal allocation must be computed.
To compute a winning agent's payment according to the exclusion compensation principle, the auctioneer must determine the social welfare that could be achieved when that winning agent is excluded from the auction.
This can be na\"ively computed by repeatedly finding the optimal allocation for each winning agent, without its participation in the auction.
Our algorithm, however, reduces the number of repetitions by using a preliminary step.
It re-computes the joint valuation function by joining the agents in reverse order to that taken when first finding the optimal allocation.
For each winning agent $j$, the joint valuation function of the rest of the agents is computed by joining the intermediate effective valuation function right before adding agent $j$, which includes agents $1,..,j-1$, and the one right before adding $j$ in the reverse order, which includes agents $j+1,..,n$.
The maximal value of this function is the maximal social welfare achievable without this agent, as required for the calculation of that agent's payment according to the exclusion compensation principle.

\subsection{Joining Two Valuation Functions}
\label{sec:joining-vals}
To na\"ively join two valuation functions, we need to find, for each possible allocation, how to best divide the resources between the two clients.
For each possible allocation of the joint agents $\vec a_j$, there are $\prod_{r=1}^R{(a_{j,r}+1)}$ possible divisions of the resource.
To compute the full joint valuation function of two clients, each with $N$ possible allocations, the number of possible resource divisions to compare is 
\begin{equation}
\sum_{\substack{\vec a_j \;s.t. \\ \vec a_j \leq \vec m}} {\left( \prod_{r=1}^R{(a_{j,r}+1)} \right)}
= \prod_{r=1}^R{\frac{m_r(m_r+1)}{2}}
= O(N^2) \text{,}
\end{equation}
for four resources, each with 15 units, $N^2 = 2^{16}$.
This number of comparisons will take a few seconds to compute on a standard CPU for each joining of two valuation functions.
For many clients, however, this can add up to a full hour.

The complexity of finding the optimal allocation and the payments depends on the complexity of joining two valuation functions.
Let $\text{J}(N)$ denote the complexity of joining two valuation functions with $N$ possible allocations.
Then the algorithm's time complexity is $O(n \cdot \text{J}(N))$.

We can reduce the complexity of $J(N)$ by reducing the number of compared allocations.
To do so, we filter out allocations that cannot maximize the social welfare.
If an allocation globally maximizes the social welfare, then
(1) it is \emph{Pareto efficient}:
one agent's allocation cannot be improved without hindering another's,
and (2) it is also a \emph{local optimum}:
the aggregated valuation cannot be increased by taking a resource from one agent and giving it to another.

Formally, the Pareto efficiency property means that if the allocation is optimal,
any left partial derivative of any single agent's valuation function is positive: $\partial_{r-} V_i(\vec a_i) > 0$.
The local optimum property means that for an optimal allocation,
any right partial derivative of any single agent's valuation function is no greater than any of the other agents' left partial derivatives: $\partial_{r+} V_i(\vec a_i) \leq \partial_{r-} V_j(\vec a_j)$.
Both are true element-wise for each resource ($r$) dimension.
Since our domain is discrete, partial derivatives are not defined.
We will define the left/right partial derivatives as the difference in the values between adjacent points in the allocation space ($dr=1$ for all the resources).

Using these properties, we restrict the search during the joining of two valuation functions.
We first eliminate client allocations in which the left partial derivative of their valuation function in one of the resource dimensions is non-positive.
Second, for each possible allocation of the first valuation function, we only consider allocations of the second function in which the condition on the partial derivative is maintained.
To accommodate \emph{boundary allocations} (allocations that reside on the valuation function's domain boundary), where the left or right partial derivative is not well defined, we assign the minimal allocation (zero) a left partial derivative of infinity, and assign the maximal allocation ($m_r$ for each resource $r$) a right partial derivative of zero.
We do this because we cannot assign an agent with less than zero or more than the maximal quantity.

These two restrictions will eliminate most of the resource divisions to $O(N)$ comparisons instead of $O(N^2)$, as shown numerically in Section~\ref{sec:complexity} and empirically in Section~\ref{sec:results}.
Algorithm~\ref{alg:joint-val} describes the joining of two valuation functions.

\begin{algorithm}[h]
 \caption{Joining two valuation functions.}
 \label{alg:joint-val}

 \SetAlgoLined
 \SetKwProg{Fn}{Function}{:}{end}
 \SetKw{Call}{call}

 \KwData{$V_i, V_j$: valuation functions}
 \KwResult{$V_r$: joint valuation function, $A_r$: the allocation that produces $V_r$}

 Initialize $V_r$ and $A_r$ to zeros\;
 Calculate $V_i$'s and $V_j$'s gradients and store them into an array of vectors\;
 Remove allocations such that $\partial_{r-} V_i(\vec a_i) \leq 0$ (for each $r$)\label{line:pareto1}\;
 Remove allocations such that $\partial_{r-} V_j(\vec a_j) \leq 0$ (for each $r$)\label{line:pareto2}\;
 \ForEach{$\vec a_i$}{ \label{line:main-loop}
	\ForEach{$\vec a_j$ such that for each $r$: $\partial_{r+} V_i(\vec a_i) \leq \partial_{r-} V_j(\vec a_j)$ and $\partial_{r+} V_j(\vec a_j) \leq \partial_{r-} V_i(\vec a_i)$ and $\vec a_i + \vec a_j \leq \vec m$}{ \label{line:local-optimum}
		$v_r \longleftarrow V_i(\vec a_i) + V_j(\vec a_j)$\;
		$\vec a_r \longleftarrow \vec a_i + \vec a_j$\;

		\If{$V_r(\vec a_r) < v_r$}{
			$V_r(\vec a_r) \longleftarrow v_r$\;
			$A_r(\vec a_r) \longleftarrow \vec a_i, \vec a_j$\;
		}
	}
 }
\end{algorithm}

\subsection{Upper-Bound Limit}
\label{sec:upper-bound-limit}

Eliminating allocations that cannot be Pareto efficient (Lines~\ref{line:pareto1} and \ref{line:pareto2} in Algorithm~\ref{alg:joint-val}) requires verifying a simple lower limit condition on the left partial derivative in the initialization of the algorithm.
The local optimum property (Line~\ref{line:local-optimum} in Algorithm~\ref{alg:joint-val}), however, requires repeated elimination for each loop iteration (Line~\ref{line:main-loop} in Algorithm~\ref{alg:joint-val}) with different multi-dimensional conditions each time.

When joining two valuation functions of agents $i$ and $j$, for each possible allocation $\vec a_i$ of agent $i$, we seek all the allocations $\vec a_j$ of agent $j$ for which the local optimum property is maintained.
Formally, we seek all $\vec a_j$ such that:
\begin{align}
\nabla_{+} V_j(\vec a_j) & \leq   \nabla_{-} V_i(\vec a_i) \\
- \nabla_{-} V_j(\vec a_j) & \leq  -\nabla_{+} V_i(\vec a_i) \\
\vec a_j & \leq \vec m - \vec a_i
\end{align}
where we define
\begin{align}
\nabla_{+} V_i(\vec a_i)
=\left( \partial_{1+} V_i(\vec a_i) , ..., \partial_{R+} V_i(\vec a_i) \right) \\
\nabla_{-} V_i(\vec a_i)
=\left( \partial_{1-} V_i(\vec a_i) , ..., \partial_{R-} V_i(\vec a_i) \right)
\end{align}
as the right and left gradients, respectively.

Each of these inequalities defines $R$ upper-bound requirements on agent $j$'s allocation, for a total of $3R$ requirements.
For each of agent $i$'s possible allocations, we need to efficiently find agent $j$'s allocations that match these requirements.
To do so, we preprocess agent $j$'s valuation function using a dedicated upper-bound data structure that allows efficient retrieval of allocations that match these requirements.
We map each possible allocation of agent $j$ ($\vec a_j$) to a new $3R$-dimensional vector:
\begin{equation}
\label{eq:vector}
\left(\nabla_{+} V_j(\vec a_j), -\nabla_{-} V_j(\vec a_j), \vec a_j \right) \text{.}
\end{equation}
We store these vectors in a \emph{$k$-dimensional upper-bound data structure}, where $k=3R$.
The data structure will contain a total of $N$ vectors and thus its complexity will depend on $k$ and $N$.
Then, for each possible allocation of agent $i$ ($\vec a_i$), we query all the vectors (defined in Equation~\ref{eq:vector}) that are smaller than or equal to the following vector:
\begin{equation}
\left(\nabla_{-} V_i(\vec a_i), -\nabla_{+} V_i(\vec a_i), \vec m - \vec a_i \right) \text{.}
\end{equation}

The $k$-dimensional upper-bound data structure must support the following methods:
\begin{itemize}
    \item \texttt{construct(all vectors)}: create the data structure.
    \item \texttt{query(vector)}: find all the vectors that are smaller than or equal to a certain vector (element-wise).
    \item \texttt{fetch()}: return all the vectors that match the last query.
\end{itemize}

We consider the $k$-dimensional ($k$-d) binary search trees
that are listed in Table~\ref{table:data-structures} along with their space and time complexities.
The complexity of result fetching is linear with the number of returned vectors and not with the number of matching vectors, because some data structures trade accuracy for efficiency, returning false positives.
\begin{table}[h]
\caption{Upper-bound data structure comparison.}
\label{table:data-structures}
\centering
\begin{tabular}{l | l l l}
                           & \multicolumn{3}{c}{Complexity ($O(...)$)}  \\ \hline
Data Structure             & Construct        & Query           & Space \\ \hline

$k$-d Tree~\cite{lueker1978data}
                           & $N\log^{k-1}{N}$ & $\log^{k-1}{N}$  & $N\log^{k-1}{N}$ \\ \hline

Simultaneous $1$-d Bin. Searches
                           & $kN\log{N}$      & $k\log{N}$       & $kN$             \\ \hline

Simultaneous $2$-d Trees
                           & $kN\log{N}$      & $k\log{N}$       & $kN\log{N}$      \\ 
\end{tabular}
\end{table}

\subsubsection{$k$-d Tree}
Algorithm~\ref{alg:lueker-construct} describes the construction of this tree.

\begin{algorithm}[h]
 \caption{$k$-d Tree Construction}
 \label{alg:lueker-construct}

 \SetAlgoLined
 \SetKwProg{Fn}{Function}{:}{end}
 \SetKw{Call}{call}

 \KwData{$v$: an array of $N$ vectors}
 \KwResult{a $k$-d Tree}

 \Call{RecursiveBuild($1$, $N$, $v$)}\;
 
 \Fn{RecursiveBuild($d$: sort dimension, $M$: array size, $v$: array of $M$ vectors)}{
  Sort the array of vectors ($v$) by dimension $d$\;
  \If{$d$ < $k$} {
	Create $\log M$ copies of the input array\;
	Partition each copied array $t=1..\log M$ into $2^{t-1}$ even parts\;
	\ForEach{$v_i$: array partition}{
		\Call{RecursiveBuild($d+1$, $v_i$, $v_i$ size)}\;
	}
  }
 }
\end{algorithm}

Figure~\ref{fig:binary-search-tree} shows an example of a four-dimensional binary search tree.
Each letter in the example represents a four-dimensional vector.
The initial array (d1) is sorted by the first dimension.
Each of the following blocks (d2, d3, d4) is built from a sorted array created on the previous block.
In this example, we partition the array up to partitions of the size of two, as creating a sub-tree of one vector is not useful.

\begin{figure}[h]     \centering
    \includegraphics[width=0.7\linewidth]{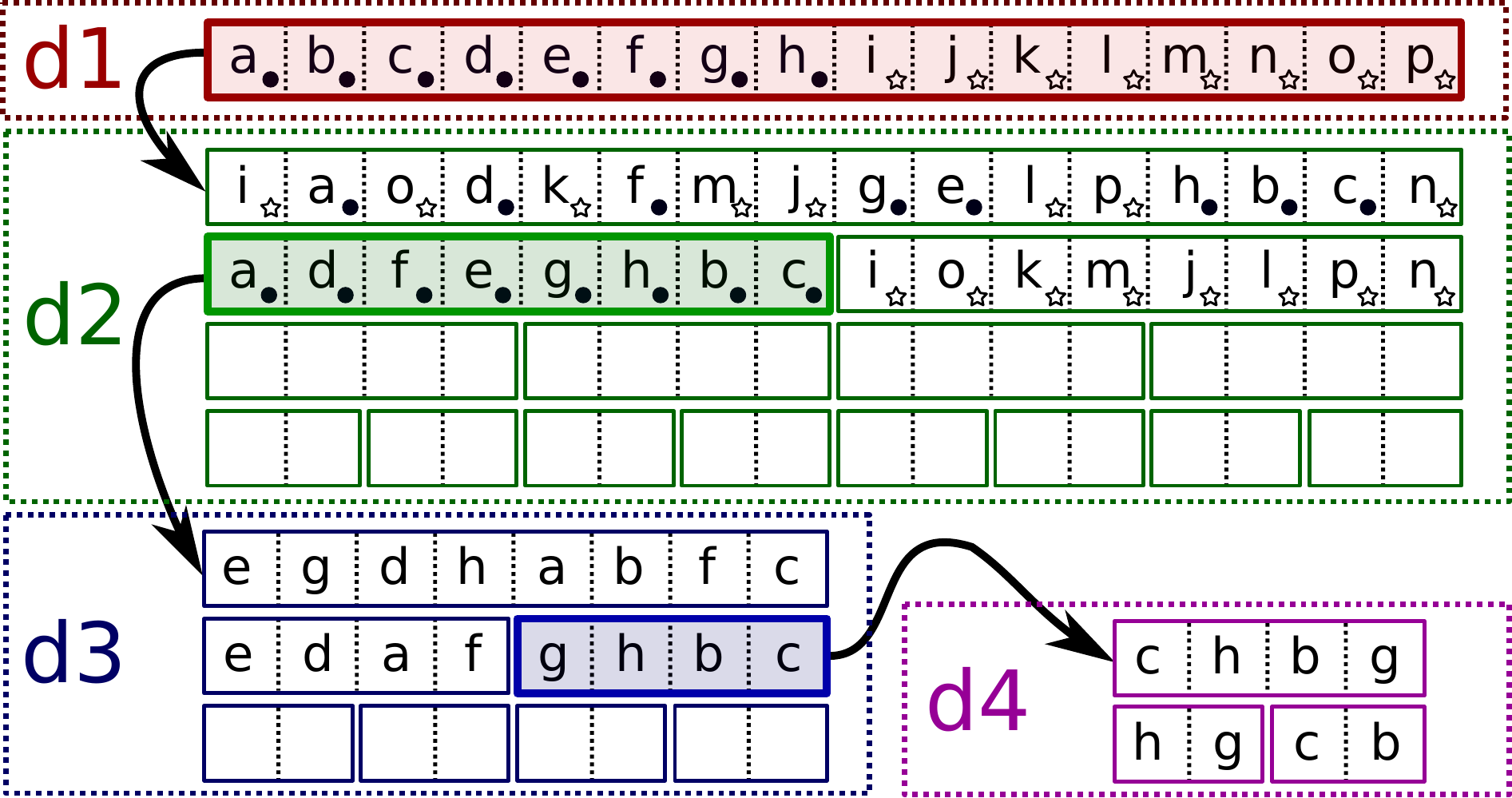}
    \caption{Example of a four-dimensional binary search tree.
    Each letter represents a four-dimensional vector.
    Each array in a rectangle marked with \emph{dX} is sorted by dimension X.}
    \label{fig:binary-search-tree}
\end{figure}

To query, we do a binary search by the first dimension on the first sorted array.
Each time the binary search continues to the upper half of the array---i.e., all the vectors in the lower half are smaller in that dimension than the query---the vectors in the lower half are filtered by the next dimension, by recursively running the query on the sub-tree created from the lower half of the array.
Then, the search continues to the upper half.
Finally, in the deepest sub-tree, we simply return all the vectors that are lower than the position returned by the binary search.
For example, starting from array d1 in Figure~\ref{fig:binary-search-tree},
if the query is larger than vector $h$, the binary search will continue to the upper half of the array ($i$ to $p$) and recursively run the query on the array with the bold frame in block d2.

This data structure never returns false positives but has prohibitive time and space complexity.
For example, four resources require a 12-dimensional data structure
with\break $O(N \log^{11}(N))$ memory and time complexity.
Even for a small $N$, e.g., $N=1024$, this can consume an entire machine.
Thus, we did not test the performance of this data structure.
Following are more efficient methods that reduce the complexity by reducing the accuracy of the results.

\subsubsection{Simultaneous $1$-d Binary Searches}
We store $k$ arrays, each sorted according to another dimension.
For each upper-bound query, we perform a \emph{simultaneous binary search} on all the arrays.
That is, instead of searching one array at a time, we perform each step on all the sorted arrays simultaneously.
In each step, some array searches continue to the lower half and some to the upper half.
We continue searching only with the array searches that continue to the lower half.
If all of the searches continue to the upper half, we continue with all of them.
When the search finishes, we have found the dimension that filters the most vectors independently of the other dimensions.
We will return all the vectors that are lower than the position the search found.
This is time and space efficient, but yields many false positives, because we only filter by one dimension.

\subsubsection{Simultaneous $2$-d Trees}
A multidimensional binary search tree problem can be relaxed by constructing many two-dimensional binary search trees, each sorted according to a different combination of two dimensions.
Then, all of them can be queried, and the vectors fetched from the tree whose query returned the least vectors.

To make the multiple queries more efficient, we use only a subset of the combinations that we believe will filter the most vectors:
all the combinations of two dimensions that originate from the same resource $r$:
\begin{equation}
\left( \partial_{r+} V_i(\vec a_i), -\partial_{r-} V_i(\vec a_i), a_{i,r} \right) \text{.}
\end{equation}
We also construct the trees in a way that reduces the number of repeated queries:
we first build $3R$ arrays, each sorted by a different dimension.
Then, from each array, we build two trees, each for the other two dimensions that originated from the same resource.

Following this, an upper-bound query is implemented:
(1) First, a simultaneous $1$-d binary search is performed on each sorted main dimension array.
The partitions that had to be searched when the binary search continued to the upper half are stored for later.
When the search is finished, the main dimension that found the lowest upper bound is chosen (or one of them is chosen if more than one remained).
(2) Next, each stored partition is searched simultaneously in the two sub-trees of the chosen main dimension. 
For each simultaneous search in the stored partition, we will return the results from the one that yielded the fewest vectors.

The query, as described above, will require one simultaneous binary search on $3R$ arrays, then at most $\log N$ simultaneous searches on two arrays.
This results in considerably fewer searches than when searching each combination of two dimensions individually.
Consequently, the time and space complexity of simultaneous $2$-d trees are not much higher than they are for simultaneous $1$-d binary searches,
but the former yields considerably fewer false positives.

\subsubsection{Combination}
Many of the vectors were created from a boundary allocation (having maximal or zero allocation in one of the resources).
Boundary allocations have a minimal partial derivative in the direction of the boundary; hence boundary allocations are never filtered by the dimensions that correspond to the partial derivative in that direction.
We can classify vectors according to their boundary type (which domain boundaries the vector's allocation resides on), and filter each class only by the \emph{vital dimensions}: those with a higher value than the minimal.
For vectors with only one vital dimension, we use simultaneous 1-d binary searches, for those with two we use a single $2$-d tree and, for those with more, we use simultaneous $2$-d trees.
This reduces both the construction time
and the query time, as the trees are smaller and each is filtered only by vital dimensions.

\section{Correctness Proof}
\label{sec:joint-val-proof}
We prove that our algorithm produces correct results, i.e., an allocation that maximizes the social welfare.

\subsection{Notations}
We use the notations from Table~\ref{table:symbol}.
Let $P={\{i\}}_{i=1}^n$ denote the set of all agents.
We define an allocation for any subset of agents $G\subseteq P$ and for maximal quantities of allocatable goods $\vec m$ as follows:
\begin{equation}
A_G^{\vec m}=\{\vec a_i\}_{i \in G} \text{.}
\end{equation}

We denote agent $i$'s valuation for an allocation $A_G^{\vec m}$ as
\begin{equation}
V_i(A_G^{\vec m})=
\begin{cases}
V_i(A_G^{\vec m}[i]), & \text{if } i\in G \\
0,                    & \text{otherwise,}
\end{cases}
\end{equation}
where $A_G^{\vec m}[i]$ is the allocation of agent $i$.

For any subset $H \subseteq G\subseteq P$ under the allocation $A_G^{\vec m}$,
we denote by $V_H(A_G^{\vec m})$ the aggregated valuation of the agents in $H$ under this allocation,
by $\vec S_H(A_G^{\vec m})$ the sum of resources allocated to the agents in $H$ under this allocation (element wise),
and by $E_H(A_G^{\vec m})$ the subset of allocations of the agents in $H$ under this allocation.
Formally,
\begin{align}
V_H(A_G^{\vec m}) & = \Sigma_{i \in {H}} (V_i(A_G^{\vec m})) \\ 
S_H(A_G^{\vec m}) & = \Sigma_{i \in {H}} (\vec a_i) \\
E_H(A_G^{\vec m}) & = \{\vec a_i\}_{i \in {H}} 
\end{align}
The social welfare of an allocation is defined as the aggregated sum of all the agents' valuations for that allocation, i.e., $\text{SW}(A_P^{\vec m})=V_P(A_P^{\vec m})$.

An allocation $A_G^{\vec m}$ is valid if $S(A_G^{\vec m})\leq \vec m$.
A valid allocation $A_G^{\vec m}$ is optimal if it maximizes the aggregated valuation:
\begin{equation}
\forall B_G^{\vec m}: V_G(B_G^{\vec m}) \leq V_G(A_G^{\vec m}) \text{,}
\end{equation}
where $B_G^{\vec m}$ is a valid allocation.

\subsection{Supporting Lemma}
In this subsection we will prove Lemma~\ref{claim:subgroup}, which supports the use of the additive process of joining the valuations one by one.
Following (Section~\ref{sec:proof-induction}) is a proof by induction that uses Lemma~\ref{claim:subgroup} to prove the optimality of the results.

\begin{lem}
\label{claim:subgroup}
For any optimal allocation $\hat A_P^{\vec m}$ and any subset of agents $G\subseteq P$,
the allocations of the agents in $G$ are also optimal for the case where
the agents in $G$ are the only agents and the number of allocatable units is exactly the sum of
their allocations.
That is, $V_G(\hat A_P^{\vec m}) = V_G(\hat A_G^{\vec m_G})$, where $\vec m_G=\vec S_G(\hat A_P^{\vec m})$.
\end{lem}

\begin{proof}
Assume the claim is false. Then, there exists an optimal allocation $\hat A_P^{\vec m}$ 
and a subset $G\subseteq P$, such that the allocations of the agents in $G$ 
are not optimal for the case where these agents
are the only agents and the number of allocatable units is exactly the sum of
their allocations.
That is, $V_G(\hat A_P^{\vec m}) \neq V_G(\hat A_G^{\vec m_G})$, where $\vec m_G=\vec S_G(\hat A_P^{\vec m})$.
There are two cases:

\begin{case}[$V_G(\hat A_P^{\vec m}) < V_G(\hat A_G^{\vec m_G})$]
Combine $\hat A_P^{\vec m}$ and $\hat A_G^{\vec m_G}$ to 
create a new allocation $\tilde A$ such that the agents in $G$ get 
the resources they get under $\hat A_G^{\vec m_G}$, and the rest of the agents get the resources they get under $\hat A_P^{\vec m}$.
The new allocation  $\tilde A$ is valid because $\hat A_G^{\vec m_G}$ is valid, and 
$\vec S_G(\hat A_P^{\vec m})=\vec m_G \geq \vec S_G(\hat A_G^{\vec m_G})$, so 
$\vec S_P(\tilde A) \leq \vec m$.
According to the assumption, 
\begin{equation}
\label{eq:smaller-case}
V_G(\hat A_P^{\vec m})  < V_G(\hat A_G^{\vec m_G}) \text{,}
\end{equation} 
and thus
\begin{equation}
V_P(\hat A_P^{\vec m})= V_G(\hat A_P^{\vec m}) + V_{P \setminus G}(\hat A_P^{\vec m}) \text{,}
\end{equation}
which according to~\eqref{eq:smaller-case} is smaller than
\begin{equation}
 V_P(\hat A_G^{\vec m_G}) + V_{P \setminus G}(\hat A_P^{\vec m}) = V(\tilde A) \text{,}
\end{equation}
in contradiction to the optimality of allocation $\hat A_P^{\vec m}$.
\end{case}

\begin{case}[$V_G(\hat A_P^{\vec m}) > V(\hat A_G^{\vec m_G})$]
Since $\vec S_G(\hat A_P^{\vec m})=\vec m_G$, then $E_G(\hat A_P^{\vec m})$ is a valid allocation for the subset of agents $G$ with maximal allocatable resources of $\vec m_G$, and it yields a higher aggregated value than $\hat A_G^{\vec m_G}$,
in contradiction to the optimality of the allocation $\hat A_G^{\vec m_G}$.
\end{case}
\end{proof}

\subsection{Proof by Induction}
\label{sec:proof-induction}
Our algorithm joins valuations into an accumulated valuation one by one.
At each step, for each number of resources $\vec m$, the algorithm iterates over all possible combination of resources $\vec m_i, \vec m_j$ such that $\vec m_i + \vec m_j \leq \vec m$.
Then, for each $\vec m$, the algorithm chooses the $\vec m_i, \vec m_j$ that yielded the maximal aggregated value.
Finally, we choose $\vec m$ that yields the maximal value in the final joint valuation function.

We prove by induction that the above algorithm finds an optimal allocation.
For generality, we do not assume that the joining of the valuations is done in any particular order.
Instead, at each step, any two valuation functions might be joined to form a single effective one.

\begin{thm}
For a subset of agents $G$ and allocatable quantities $\vec m$ of goods, the algorithm finds an optimal allocation $\hat A_G^{\vec m}$.
\end{thm}

\begin{proof}

\begin{case}[$|G|=1$]
For one agent, no joining of two valuations is needed.
The algorithm simply chooses the maximal valuation for any allocation up to $\vec m$.
This is the maximum social welfare by definition.
\end{case}

\begin{case}[Inductive hypothesis]
Suppose the theorem holds when $|G|\leq k$, for some
$k \geq 1$.
Let $|G|=k+1$.

Consider any two non-empty, disjoint subsets: $X$ and $Y$, where $X\cup Y=G$.
By the pigeonhole principle, $|X|\leq k$ and $|Y|\leq k$, and
\begin{equation}
\vec S_G(\hat A_{G}^{\vec m}) = \vec S_X(\hat A_G^{\vec m}) + \vec S_Y(\hat A_G^{\vec m}) \leq \vec m
\end{equation}
since the optimality of allocation $\hat A_{G}^{\vec m}$ implies it is valid.

Let us denote
$\vec m_X=\vec S_X(\hat A_G^{\vec m})$, $\vec m_Y=\vec S_Y(\hat A_G^{\vec m})$.
Then
\begin{equation}
\vec m_X + \vec  m_Y \leq \vec m.
\end{equation}
According to Lemma~\ref{claim:subgroup},
\begin{align}
V_X(\hat A_G^{\vec m}) & = V_X(\hat A_X^{\vec m_X}) \\
V_Y(\hat A_G^{\vec m}) & = V_Y(\hat A_Y^{\vec m_Y})
\end{align}
$\Longrightarrow$ $V_X(\hat A_X^{\vec m_X}) + V_Y(\hat A_Y^{\vec m_Y}) = V_G(\hat A_G^{\vec m})$.

Hence, since we search all the options where $\vec m_X + \vec m_Y \leq \vec m$
and find optimal allocations $\hat A_X^{\vec m_X}$, $\hat A_Y^{\vec m_Y}$ for each of them, we must encounter an allocation with the above aggregated valuation.
Because it is the maximal value, our algorithm will prefer this allocation to the alternatives.
So, the theorem holds for $|G|=k+1$.
\end{case}

By induction, the theorem holds for every size of $G$.
\end{proof}

\section{Complexity Analysis of Joining Two Valuations}
\label{sec:complexity}
We first show the worst-case time complexity of $O(N^2)$, which may be relevant only in unrealistic scenarios.
Then, we analyze the worst-case complexity of a single resource over realistic valuation functions, and find it equal $O(N \log{N})$.
Finally, we show that multiple resources yield the same time complexity, but on average over all possible realistic valuation functions.

\subsection{Worst Case}
The worst case complexity of joining two valuation functions is $O(N^2)$, 
when for every query, the number of matching allocations is proportionate to $N$.
This can happen, for example, when both valuation functions are linear, with an identical slope.
Any of the $N$ queries on one of the functions will return every allocation ($O(N)$), as 
the upper-bound limit is inclusive.
This adversarial example, however, is unlikely on a real cloud, with a mixture of clients and valuation functions, and where precise linear scaling is rare.
We will thus consider in the following only strictly convex/concave functions, i.e., without any precise linear parts.

\subsection{Single Resource}

To analyze the  complexity we will assume $N \rightarrow \infty$, which approximates a smooth continuous function were the left partial derivative is equal to the right.
This reduces the local optimum property to a single rule: for an optimal allocation, all the agents' valuation functions have identical identical gradients.

For a single resource with concave/convex valuation functions, each
derivative value
is obtained at most once.
Hence, each query will match at most one allocation.
For a function with one or more inflection points, each query will match a number of allocations up to the number of inflection points in the function.
The number of inflection points is related to the number of hierarchies in the resource.
For example, a CPU might have two inflection points: when switching from a single-core to multiple-cores, and then to multiple-chips.
Memory might also have two inflection points when switching between cache, RAM and storage.
Five inflection points, however, might be considered unrealistically high for computing resource valuation functions.
Thus, we consider the number of possible inflection points for each resource to be a constant as it is independent on the parameters ($n$, $N$ and $R$) and is generally small.
This yields a maximal complexity of $O(N)$.

The time complexity of joining two valuation functions is at least $O(N \log{N})$, the data structure construction complexity.
Hence, the complexity of joining two valuation functions is $O(N \log{N})$.

\subsection{Multiple Resources}
Similarly to a single resource, for multiple resources with concave/convex valuation functions, each gradient vector is obtained at most once.
We
can consider each resource to have inflection points independently of the other resources, e.g., it is possible to switch from a single processor to a multi-processor algorithm regardless of the RAM usage.
Thus, if each resource has $t$ inflection points, we can divide the valuation function domain into $(t+1)^R$ sections, each being convex or concave.
That is, each gradient vector might be obtained at most once in each of these sections.
The actual number of matches is much lower than $(t+1)^R$, and is constant as shown in Section~\ref{sec:potential-analysis}.

We reconcile these differences by showing that the average case, over all possible realistic valuation functions yields a constant number of matching allocations.
To do this, we will assume without loss of generality that the partial derivatives on each of the inflection points and in the function boundaries distribute uniformly from zero to the maximal derivative.
The partial derivatives of the required gradient will also distribute uniformly with the same boundaries.
Then, for exactly two inflection points per resource, we will have three sections, each with different uniformly distributed boundaries.
The probability of a single derivative that is uniformly distributed to be in these boundaries is $\frac{1}{3}$, and thus, for each resource, exactly one section is expected to have this gradient.
Thus, regardless of the number of resources $R$, exactly one section is expected to have the required gradient (out of the total $(t+1)^R$).
Since only a single matching allocation exists in that section, the expected number of matching allocations is exactly one.

Furthermore, if we assume that the required gradient has different derivative boundaries, as we would expect in the real world, then a higher number of inflection points will yield a single matching section as well.
If the first client's valuation function has a maximal derivative $d$ times higher than the second, then $\lfloor 3 \cdot d - 1 \rfloor$ number of inflection points per resource will yield at most one matching allocation per query.
Since the joint valuation function is expected to have higher derivatives with each joining, we would expect $d$ to grow in each step, and thus reduce the number of matching allocations.
This yields an average complexity of $O(N)$ over realistic valuation functions.

Hence, similarly to a single resource, the complexity of joining two multi-resource valuation functions is $O(N \log{N})$.

\section{Evaluation}
\label{sec:evaluation}
Here we empirically evaluate the algorithm's complexity, and verify that our implementation is efficient enough to be applicable in a real system.

\subsection{Implementation Details}
We implemented the joint function algorithm and Maill\'e and Tuffin's~\cite{maille2004multi} algorithm in C++ and Python.
The code is available as open source\footnote{Available from: \url{https://bitbucket.org/funaro/vecfunc-vcg}.}.

The joining of two valuation functions and the upper-bound data structures were  implemented in C++.
The algorithm can accept any upper-bound data structure as a template parameter.
We implemented the na\"ive joining in C++ as well.
Both implementations accept two $R$-dimensional tensors, which represent the clients' valuation functions (or effective joint valuation functions), and return an $R$-dimensional tensor, which is the joint valuation function.
The C++ library is called (via a Python wrapper) to join the functions one by one, and the allocation and payment calculations are implemented in Python.

Our C++ implementation of Maill\'e~and~Tuffin's~\cite{maille2004multi} algorithm accepts all the clients' bids and returns the optimal allocation.
This C++ implementation is called once (via a Python wrapper) to compute the optimal allocations, and then again for each winning client to compute the payments.

\subsection{Benchmark Dataset}
\label{sec:datasets}
We considered three different types of datasets: \emph{concave}, \emph{increasing}, and \emph{mostly-increasing}.
We produced 10 datasets of each type, each with 256 clients that participate in the VCG auction.
The \emph{concave} datasets contain concave, strictly increasing valuation functions. 
These datasets are used to compare our results to Maill\'e and Tuffin's method, 
where the types of valuation functions are very restricted~\cite{maille2004multi}.
The \emph{increasing} datasets include weakly increasing valuation functions that might not be concave.
This is our main test case as real-life valuation functions may have multiple inflection points~\cite{funaro2019stochastic,ye2015rochester,cameron2014we,wilkes2009utility,zhu2006utility,lee2007precise}.
Valuation functions, however, are not expected to decrease when more resources are offered,
if these resources can be freely discarded.
The \emph{mostly-increasing} datasets include valuation functions with multiple maximum points (non-monotonic).
Such functions will increase for a large part of their input, but may occasionally decrease.
They are realistic when the hindering resources are not disposed of, as is the case,
for example, when allocating more RAM lengthens garbage collection time and performance drops~\cite{agmonben-yehuda2014ginseng,yang2006cramm}.
We use these datasets to show that our algorithm performs well even with non-monotonic functions.
We did not test strictly convex valuation functions as they are not realistic.

For each client, we produced an $R$-dimensional valuation function
($V_i: [0,1]^R\in \mathbb{R}^R \mapsto{[0,\infty)\in \mathbb{R}}$),
which it uses as its bid.
We generated $R$ intermediate single-dimensional functions
($v_i^r: [0,1]\in \mathbb{R} \mapsto{[0,1]\in \mathbb{R}}$) without loss of generality,
where an input value of $1$ represents the entire available resource $r$, and an output of $1$ represents the client's maximal valuation of the resource.

To compute a client's valuation function---i.e., its bid for each bundle of units---for each single-dimensional function, we sampled a vector sized according to the number of available units for each resource and computed the vectors' tensor product:
$V_i = v_i^1 \otimes ... \otimes v_i^R$.
This yielded an $R$-dimensional tensor with values in the range of $[0,1]\in \mathbb{R}$.
To produce a valuation function of fewer than $R$ dimensions ($0<r<R$), we used the same dataset but only with the first $r$ intermediate single-dimensional functions.

We modeled the clients' maximal valuations using data from Azure's public dataset~\cite{cortez2017resource}, which includes information on Azure's cloud clients, such as the bundle rented by each client.
Assuming the client is rational, the cost of the bundle is a lower bound on the client's valuation of this bundle.
We modeled the clients' expected revenue using a Pareto distribution (standard in economics) with an index of $1.1$.
A Pareto distribution with this parameter translates to the 80-20 rule: 20\% of the population has 80\% of the valuation, which is reasonable for income distributions~\cite{souma2001universal}.

For each client, we drew a value from this Pareto distribution, with the condition that the value is higher than the client's bundle cost (i.e., a conditional probability distribution).
We then multiplied each client's $R$-dimensional tensor with the maximal value drawn from the Pareto distribution, to produce the client's valuation function.

\subsection{Experimental Setup}
We evaluated our algorithm on a machine with 16GB of RAM and two Intel(R) Xeon(R) E5-2420 CPUs @ 1.90GHz with 15MB LLC. 
Each CPU had six cores with hyper-threading enabled, for a total of 24 hardware threads.
The host ran Linux with kernel 4.8.0-58-generic \#63\texttildelow{}16.04.1-Ubuntu.
To reduce measurement noise, we tested using a single core, leaving the rest idle.

\section{Results}
\label{sec:results}

The combination of data structures was chosen for the purpose of the evaluation as it performed the best.
This is shown in the data structure comparison in Section~\ref{sec:ds-analysis}.

Our algorithm scales linearly to the number of possible allocations ($N$), for any number of resources, as depicted in Figure~\ref{fig:killer}.
Although the performance differences between the concave, increasing and mostly-increasing datasets were insignificant, we can see that our algorithm performs better on the mostly-increasing dataset.
This is because more allocations were eliminated in the preprocessing phase due to their negative left partial derivative.
This preprocessing was included in the algorithm's runtime.

\begin{figure}[t]	\centering

	\hfill
    \subfloat[Runtime of each joining.\label{fig:killer-join-time}]{    	\includegraphics[width=0.45\linewidth,valign=t]{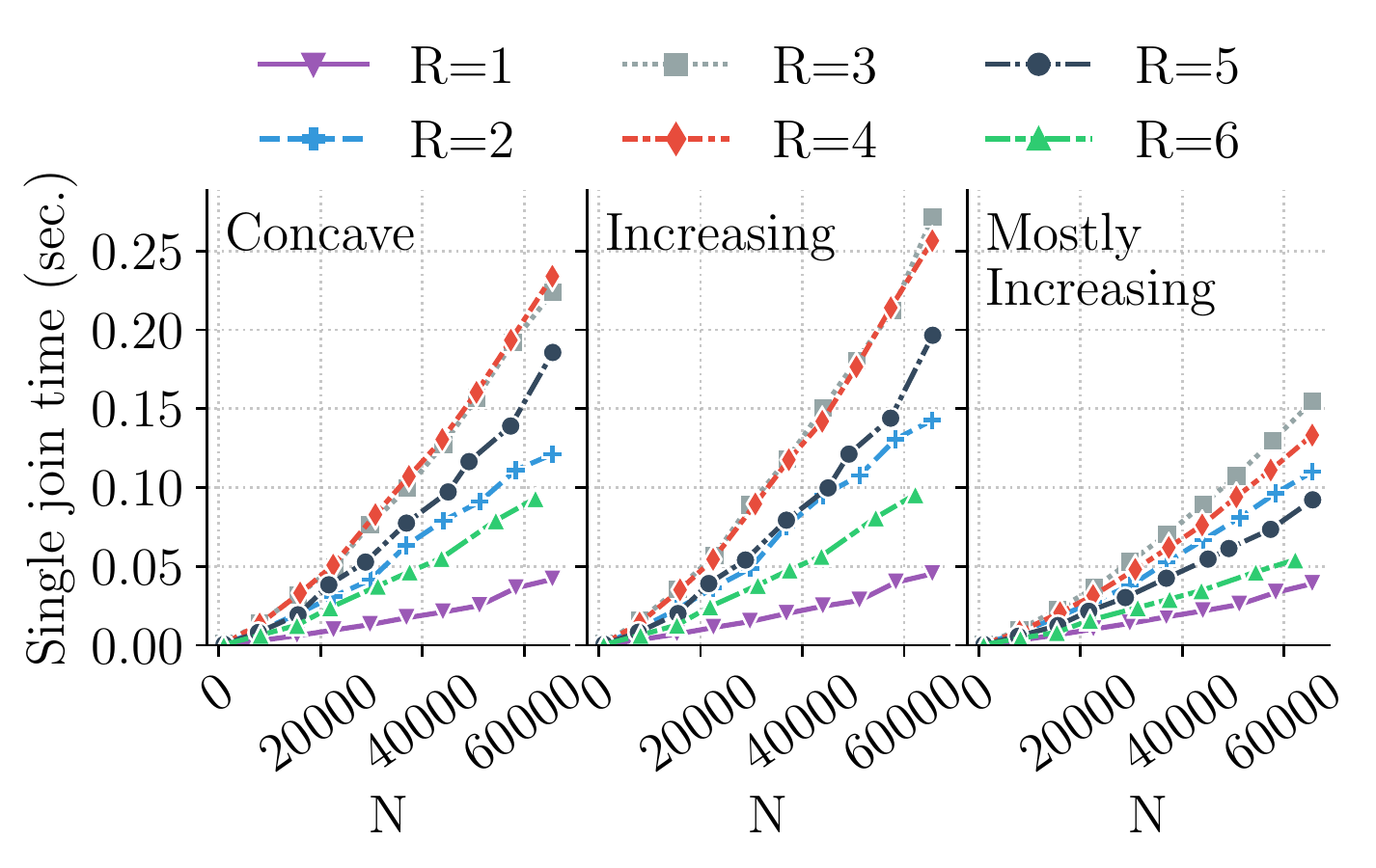}
    }
    \hfill
    \subfloat[Runtime of a full auction and payment calculation for 256 clients.\label{fig:killer-auction-time}]{    	\includegraphics[width=0.45\linewidth,valign=t]{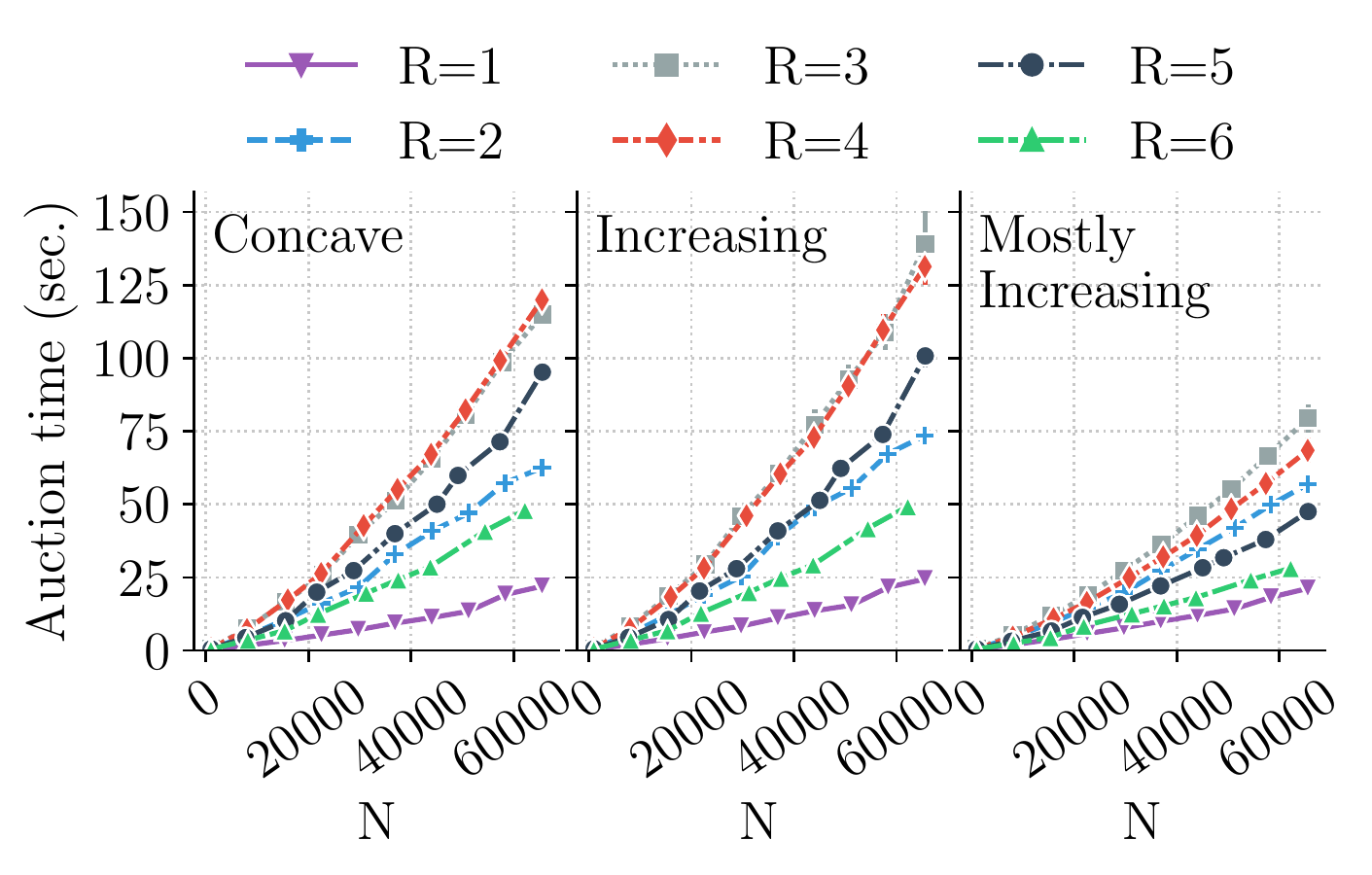}
    }
    \hfill
    \caption{The performance of our algorithm in each of our datasets (concave, increasing and mostly-increasing).}
    \label{fig:killer}
    \end{figure}

Adding resources results in larger vectors and thus higher complexity; at the same time, more vectors are eliminated in the preprocessing phase.
This is why we see an increase in runtime for up to four resources, after which the performance begins to improve.

Figure~\ref{fig:killer-auction-time} shows that the multi-resource auction is feasible even in the worst case: for concave/increasing valuation functions, and for three and four resources with 256 clients, a full auction takes less than two minutes for over 60,000 possible allocations.

\subsection{Na\"ive Joining of Valuation Functions}
The results show (Figure~\ref{fig:brute-force}) that the performance of the na\"ive approach for joining two valuation functions fits the expected curve, as shown in Section~\ref{sec:joining-vals}, for any number of resources.
Figure~\ref{fig:brute-force} depicts the performance for the increasing dataset.
The na\"ive joining is not affected by valuation function properties such as monotonicity.
The complexity function, described in Section~\ref{sec:joining-vals}, passes through all the markers, i.e., fits the actual performance perfectly.
Each line, however, had to be scaled by a different factor to fit the markers.
This might be an effect of the cache prefetching combined with the C-style multidimensional array representation.
The na\"ive joining compares each allocation $\vec a_i$ to all allocations $\vec a_j \;s.t.\; \vec a_j \leq \vec m - \vec a_i$.
For multidimensional valuation functions that are represented as C arrays, we will read the array non-continuously when $\vec a_i > 0$.
This will reduce the effectiveness of the cache prefetching as it relies on the continuity of the reading.

\begin{figure}[h]	    \centering
    \includegraphics[width=0.5\linewidth]{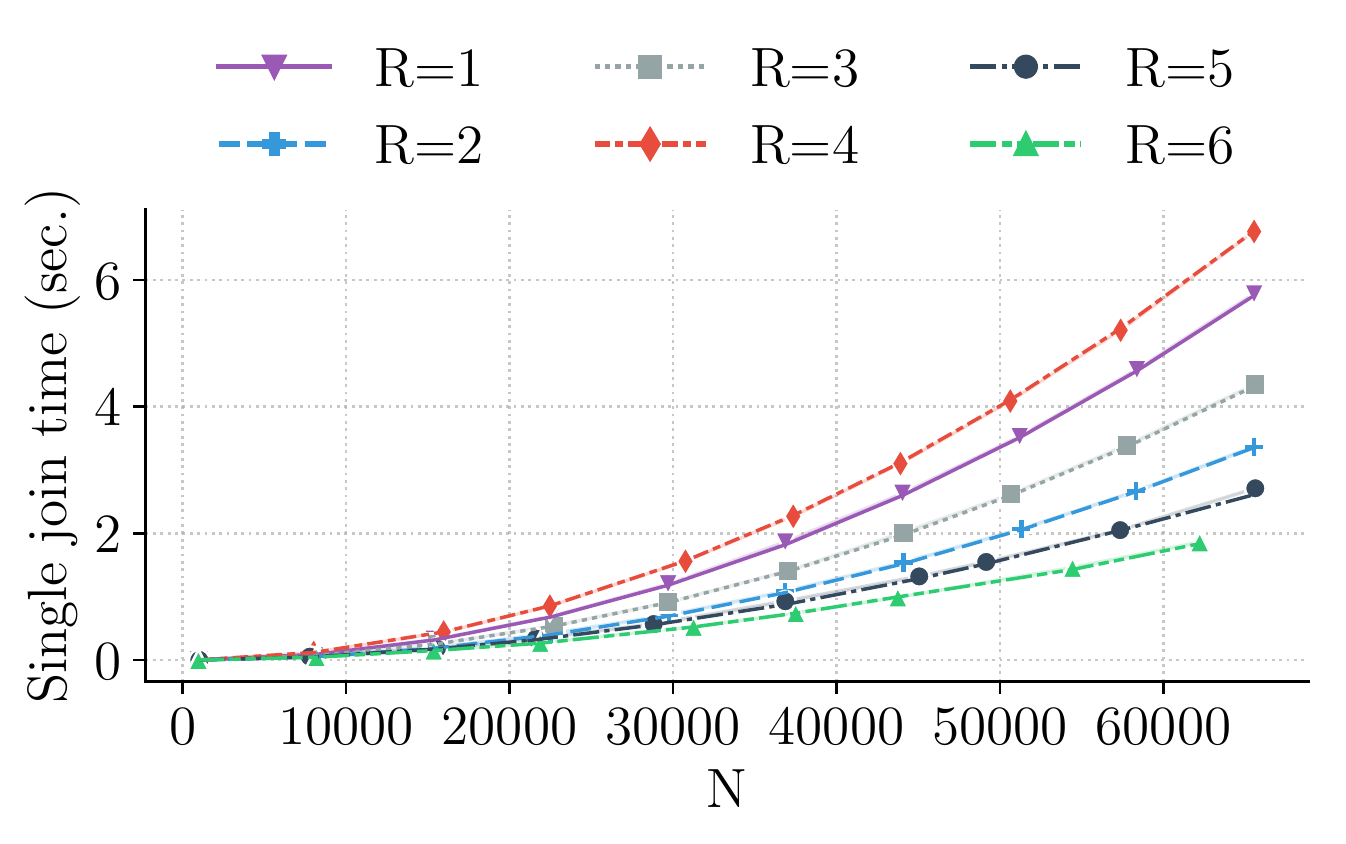}
    \caption{The performance of the na\"ive approach for joining two valuation functions.
    Markers depict the performance with different numbers of resources.
    The lines are the complexity function described in Section~\ref{sec:joining-vals}, scaled to fit the markers.}
    \label{fig:brute-force}
    \end{figure}

\subsection{Ideal Case Analysis}
\label{sec:potential-analysis}
We ran another set of experiments on each dataset, where we 
counted, in each joining of two valuation functions, the number of allocations that matched the queries of the one valuation function, for each allocation of the other.
Figure~\ref{fig:ds-compared-points} shows the results.
The number of matching allocations converges to a constant number.
Thus, were we to have an ideal data structure that does not return false positives and with reasonable query and construction time, the complexity of joining two valuation functions would be
$O(N)$.

\begin{figure}[h]    \centering
    \includegraphics[width=0.89\linewidth,valign=t]{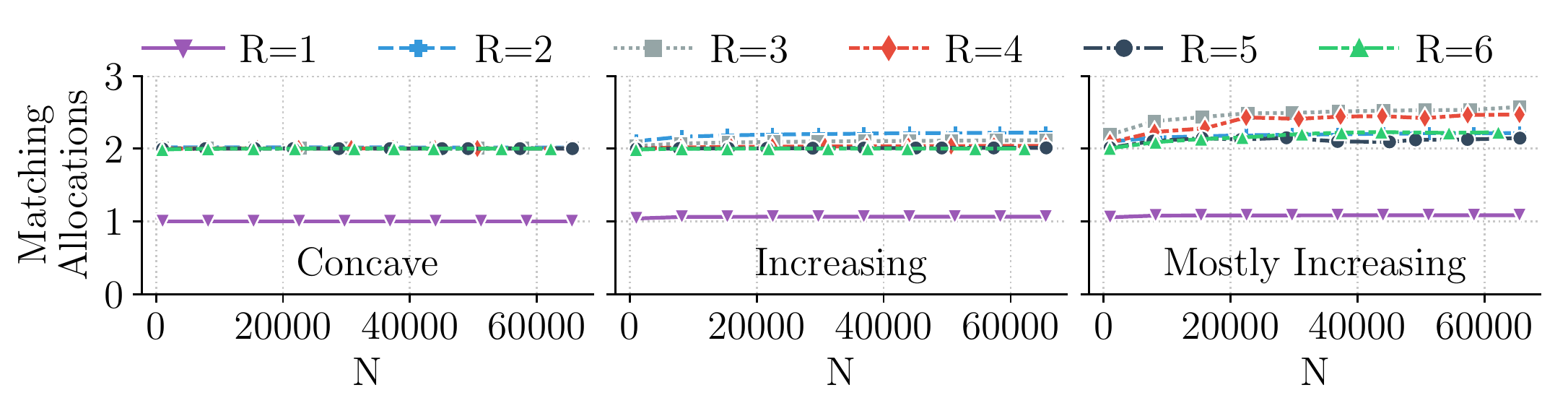}
    \caption{The average number of matching allocations for each query for each dataset and number of resources.
    }
    \label{fig:ds-compared-points}
\end{figure}

\subsection{Data Structure Analysis}
\label{sec:ds-analysis}
We timed each step of the algorithm: creating the data structure, performing the query and fetching the allocations.
The results are shown in Figure~~\ref{fig:ds-analysis}.

Simultaneous $1$-d binary searches have the fastest construction, but their false-positive ratio grows quickly with $N$, as indicated by the longer fetching time.
Hence they are not scalable.

Simultaneous $2$-d binary search trees performed similarly to the combination of trees.
The construction time of the latter is better as some trees contain fewer vectors.
We, therefore, recommend this solution.

In all the tested cases and for all the data structures, the construction time is more than 30\% of the total runtime,
and over 70\% of the total runtime in some cases.
Further improvement could be obtained by finding a data structure with a smaller construction time.
If we consider simultaneous $1$-d binary searches as a lower bound on the construction time---an upper-bound data structure must at least sort the vectors by each dimension---then improving the construction time will improve the algorithm by 10\%-20\% at most.

For similar reasons, the query time could not be lower than for $1$-d binary searches, which is nearly identical to the combination of data structures.
Hence, no further improvement in the query time is expected.

The fetching time is 10\%-25\% of the total time for the combination of data structures.
Reducing the number of false positives may reduce this phase's time.
Thus, any improvement of the algorithm by increasing the accuracy of the data structures is bounded by 10\%-25\%.

Fetching the allocations includes an additional filtering (one by one) to remove all the false positives.
Thus, the performance of the final step---i.e., comparing the allocations---is independent of the data structure.
From the above, we conclude that any speedup of the algorithm via an improved data structure is limited by 20\%-50\%.

\begin{figure}[t]    \centering
    \includegraphics[width=\linewidth]{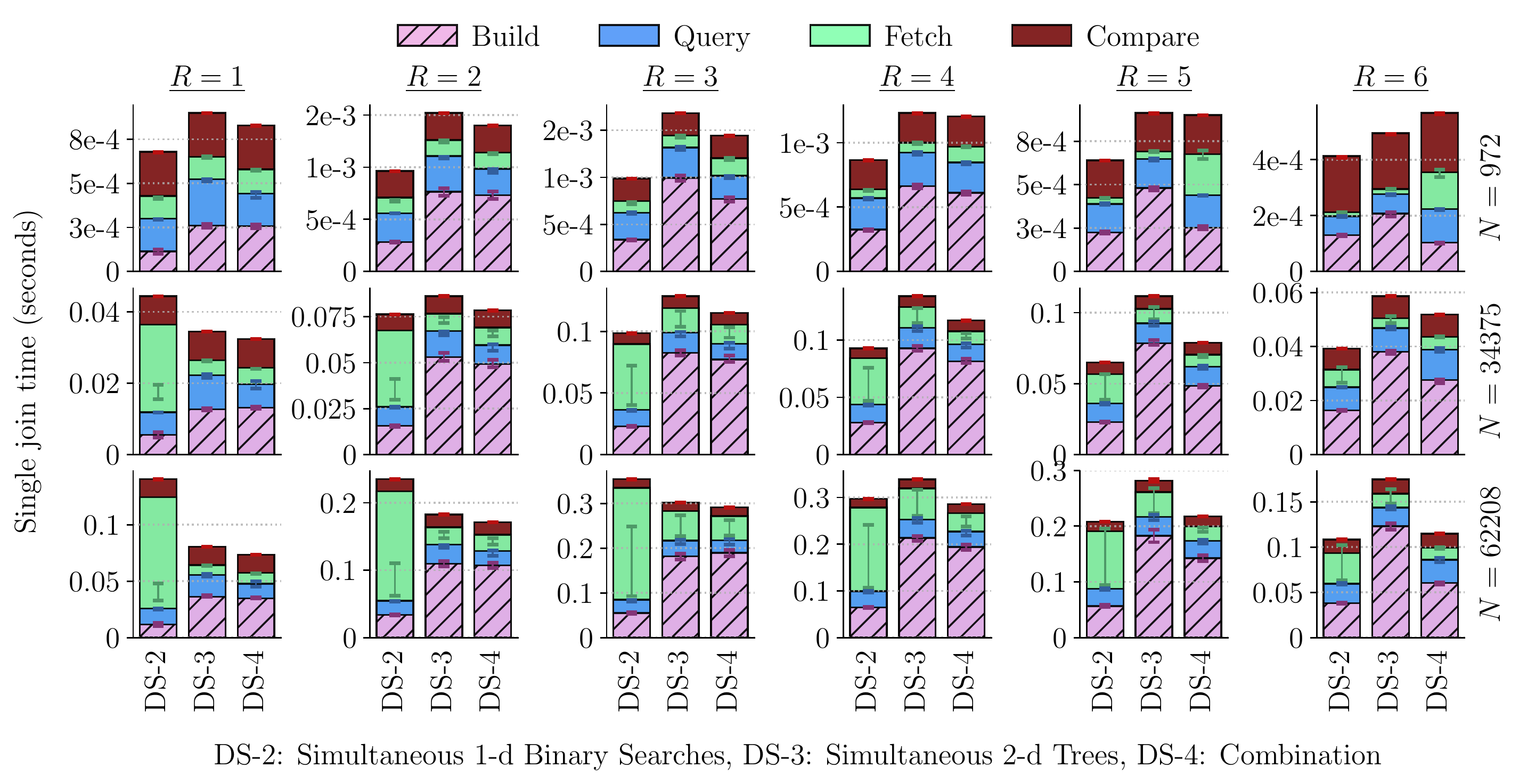}
    \caption{The runtime portion of each function of data structures (DS) when used by our algorithm in an auction.
    The whiskers represent the \nth{25} and \nth{75} percentiles.}
    \label{fig:ds-analysis}
    \end{figure}

\subsection{False Positives}
We measured the ratio of false positive results when applying our algorithm on all the datasets (Figure~\ref{fig:ds-false-positive}).
For any of the resources, the false-positive ratio grows linearly with the number of possible allocations.
Using an ideal data structure could reduce the number of false positives by up to a factor of 60.
Nonetheless, such an improvement could only speed up the optimization by 10\%-25\%, as shown in Section~\ref{sec:ds-analysis}.

\begin{figure}[h]	    \centering
    \includegraphics[width=0.5\linewidth]{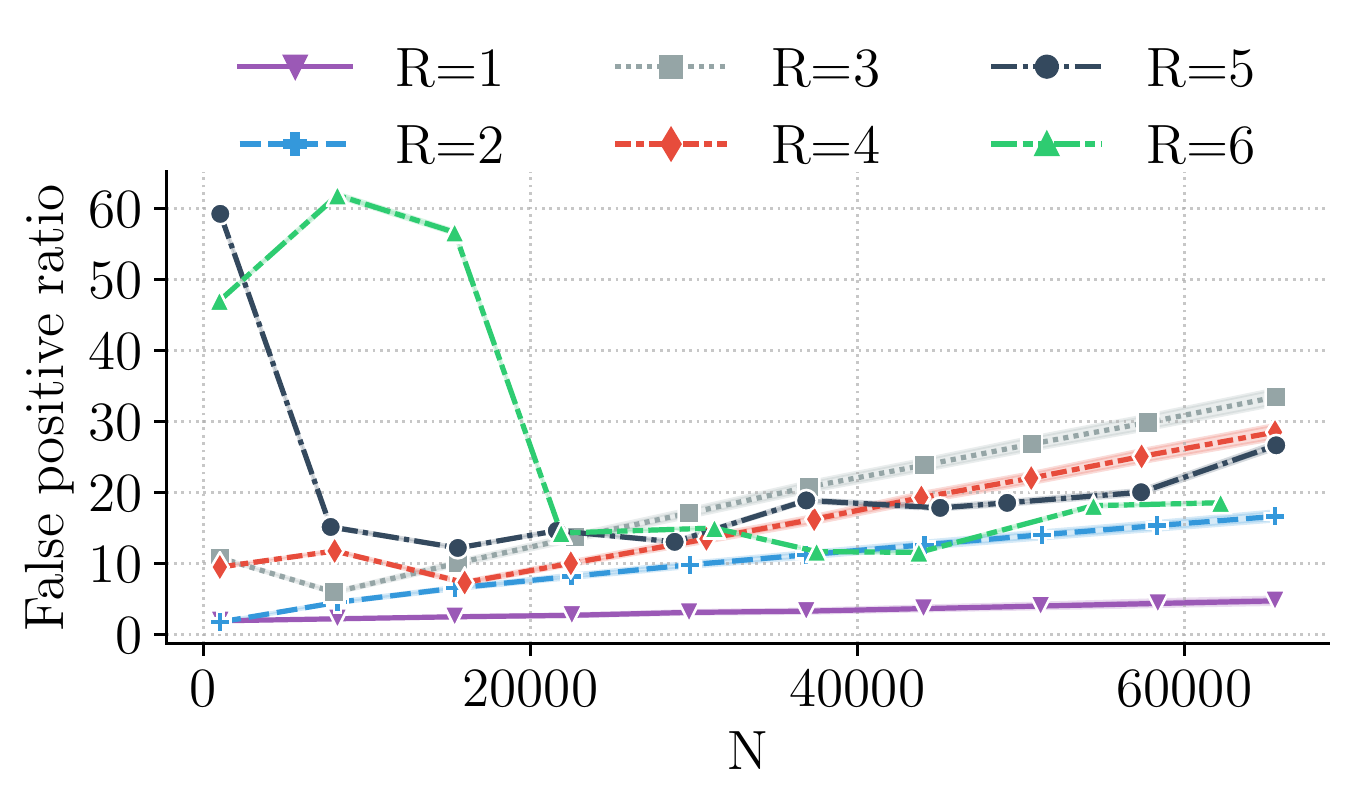}
    \caption{False-positive ratio for the combination of data structures for each number of resources.
    The translucent bands represent the confidence intervals.}
    \label{fig:ds-false-positive}
    \end{figure}

\subsection{Separate Single-Resource Auction}
We compared our multi-resource VCG auction implementation to the alternative of performing an auction for each resource separately.
We used Maill\'e and Tuffin's method for a single-resource auction with the concave valuation functions dataset.
For each resource $r$, each client bid its intermediate single-dimensional valuation functions $v_i^r$ (see Section~\ref{sec:datasets}).
Each client's maximal valuation was treated as a budget, which was partitioned equally among its valuation functions for each resource.
For example, for two resources, a client with a maximal valuation of 10 would have a maximal valuation of 5 for each of its resources.

Such an approach reduces the social welfare by over 60\% on average compared to the optimum for two resources (Figure~\ref{fig:multi-single-resource}).
When more resources are auctioned, the social welfare decreases even further.

\begin{figure}[h]    \centering
    \includegraphics[width=0.6\linewidth,valign=t]{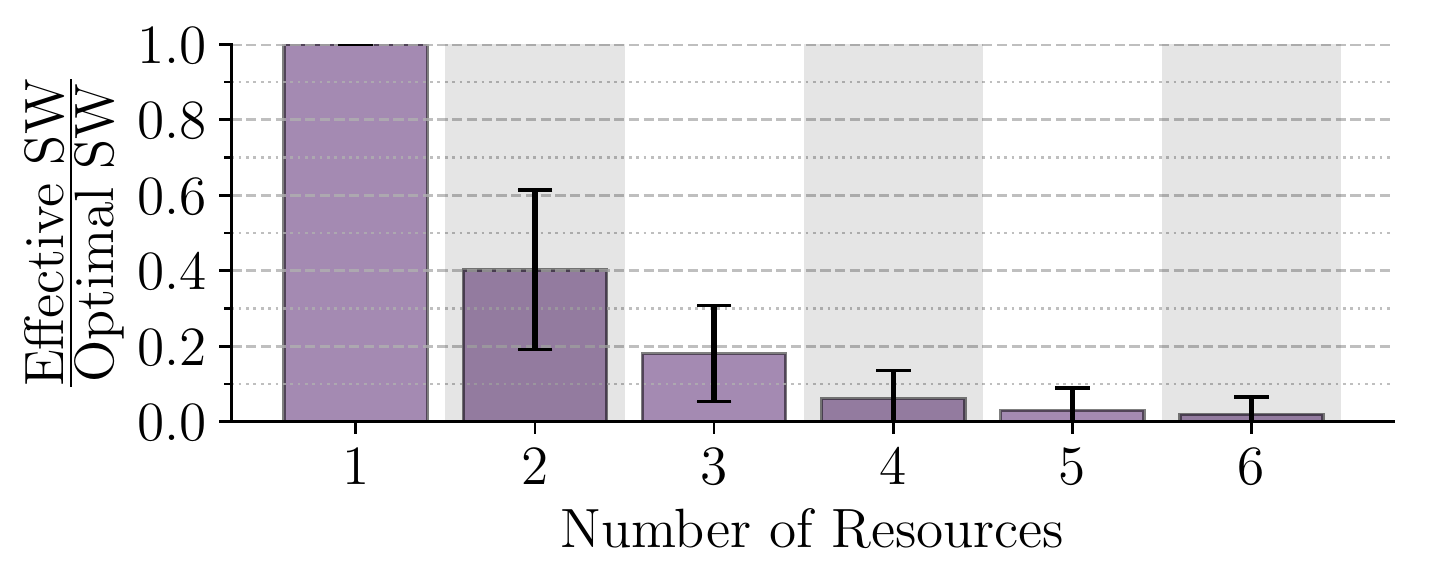}
    \caption{The social welfare when using a separate single-resource auction normalized to the optimal social welfare.
    The whiskers represent the standard deviation.}
    \label{fig:multi-single-resource}
\end{figure}

\subsection{Verification}
To verify our implementation, we compared our algorithm's results with those of Maill\'e and Tuffin~\cite{maille2004multi} 
using the concave dataset and a single resource.
For all the tested numbers of units ($N$), our algorithm produced the same allocation and payments as Maill\'e and Tuffin's method.

We also compared our algorithm's results for two and more resources to those of the na\"ive implementation.
For all the tested numbers of units ($N$) and resources ($R$), our algorithm produced identical results to the na\"ive implementation.

\section{Related Work}
The Resource-as-a-Service (RaaS) cloud~\cite{agmonben-yehuda2012resource} is a vertically elastic cloud model that allows providers to rent adjustable quantities of individual resources for short time intervals---even at a sub-second granularity.
It deploys economic mechanisms to allocate the resources quickly and efficiently.
The RaaS model was implemented in Ginseng: 
first, to allocate resources
for RAM~\cite{agmonben-yehuda2014ginseng} using a VCG-like auction mechanism, and later for last-level-cache~\cite{funaro2016ginseng} using a full VCG auction.

Many solutions were suggested for allocating multiple resources in the cloud.
Non-economic solutions may optimize fairness according to clients' requirements~\cite{ghodsi2011dominant,dolev2012no,gutman2012fair,skowron2013non,hines2011applications} or consider the clients as a black box and use host measurements instead~\cite{xiao2013dynamic}.
Hadi~et~al.~\cite{goudarzi2011multi} aim to maximize the profit of the providers by meeting client's SLA.
Some achieve truthfulness under restrictive conditions on the types of clients allowed to participate in the auction~\cite{lazar1999design,maille2004multi,agmonben-yehuda2014ginseng,bae2008efficient,mualem2008truthful}, or by restrictions on the bidding language~\cite{yang2007vcg,cai2012optimal,jia2009revenue,fonseca2017faircloud}.
Other solutions offer only near-optimal auction results~\cite{sanghavi2004optimal,nisan2007computationally,zhang2014dynamic,fukuta2011vcg,archer2004approximate}.

\section{Conclusions and Future Work}
We introduced a new efficient algorithm to allocate multiple divisible resources via a VCG auction, without any restrictions on the valuation functions.
We proved the algorithm's correctness, verified it 
experimentally, and showed its efficiency on a large number of resources and its scalability when increasing the number of units per resource.

We analyzed how the different properties of the valuation functions affect the algorithm's performance.
We showed that using only concave valuation functions negligibly decreases the complexity compared to increasing valuation functions, and that mostly-increasing ones perform the best.

We combined data structures, tailoring them to our input data to create a data structure that produces fewer false positives and has faster construction time.
We analyzed different data structures and showed a potential speedup of up to $2\times$.
Finding a better upper-bound data structure is left for future work.

Our algorithm allows cloud providers to implement the RaaS~\cite{agmonben-yehuda2012resource} model.
They can deploy a multi-resource auction for allocation of additional resources in an existing VM every two minutes for up to 256 clients in a single physical machine.
Our implementation can be adapted simply to use succinct valuation functions that are only defined on a small subset of the allocations.
This will eliminate the exponential factor of $N$ in $R$, the number of resources,
which may greatly improve the performance and might allow a sub-second auction granularity for a large number of clients.
A succinct implementation might also support continuous valuation functions with good performance but unbounded complexity.
Adapting the implementation for continuous succinct valuation functions is left for future work.

\section{Acknowledgment}
We thank
Deborah Miller, Sharon Kessler, Hadas Shachnai, Tamar Camus, Ido Nachum, Danielle Movsowitz and Shunit Agmon
for fruitful discussions.
This work was partially funded by the Hasso Platner Institute, and by the
Pazy Joint Research Foundation.

\bibliographystyle{splncs04}
\bibliography{funaro-bib/liranfunaro}

\begin{thebibliography}{10}
\providecommand{\url}[1]{\texttt{#1}}
\providecommand{\urlprefix}{URL }
\providecommand{\doi}[1]{https://doi.org/#1}

\bibitem{agmonben-yehuda2012resource}
Agmon Ben-Yehuda, O., Ben-Yehuda, M., Schuster, A., Tsafrir, D.: The
  resource-as-a-service ({RaaS}) cloud. In: Proceedings of the 4th USENIX
  Conference on Hot Topics in Cloud Computing (HotCloud). USENIX Association
  (2012)

\bibitem{agmonben-yehuda2013deconstructing}
Agmon Ben-Yehuda, O., Ben~Yehuda, M., Schuster, A., Tsafrir, D.: Deconstructing
  {Amazon} {EC2} spot instance pricing. ACM Transactions on Economics and
  Computation (TEAC)  \textbf{1}(3) (2013). \doi{10.1145/2509413.2509416}

\bibitem{agmonben-yehuda2014ginseng}
Agmon Ben-Yehuda, O., Posener, E., Ben-Yehuda, M., Schuster, A., Mu'alem, A.:
  {Ginseng}: Market-driven memory allocation. In: Proceedings of the 10th ACM
  SIGPLAN/SIGOPS International Conference on Virtual Execution Environments
  (VEE). vol.~49. ACM (2014)

\bibitem{akbar2006solving}
Akbar, M.M., Rahman, M.S., Kaykobad, M., Manning, E.G., Shoja, G.C.: Solving
  the multidimensional multiple-choice knapsack problem by constructing convex
  hulls. Computers \& Operations Research  \textbf{33}(5),  1259--1273 (2006)

\bibitem{alibabaAlibabaSpot}
Alibaba: Alibaba cloud spot instances.
  \url{https://www.alibabacloud.com/help/doc-detail/52088.htm} (2018),
  accessed: 2018-05-03

\bibitem{AmazonBurstableInstances}
Amazon: Amazon {EC2} burstable performance instances.
  \url{https://aws.amazon.com/ec2/instance-types/\#burst} (2018), accessed:
  2018-07-25

\bibitem{AmazonSpotInstances}
Amazon: Amazon {EC2} spot instances.
  \url{https://aws.amazon.com/ec2/spot/details/} (2018), accessed: 2018-07-25

\bibitem{archer2004approximate}
Archer, A., Papadimitriou, C., Talwar, K., Tardos, E.: An approximate truthful
  mechanism for combinatorial auctions with single parameter agents. Internet
  Mathematics  \textbf{1}(2),  129--150 (2004).
  \doi{10.1080/15427951.2004.10129086},
  \url{http://dx.doi.org/10.1080/15427951.2004.10129086}

\bibitem{bae2008efficient}
Bae, J., Beigman, E., Berry, R., Honig, M.L., Vohra, R.: An efficient auction
  for non concave valuations. In: 9th International Meeting of the Society for
  Social Choice and Welfare (2008)

\bibitem{cai2012optimal}
Cai, Y., Daskalakis, C., Weinberg, S.M.: Optimal multi-dimensional mechanism
  design: Reducing revenue to welfare maximization. In: Foundations of Computer
  Science (FOCS), 2012 IEEE 53rd Annual Symposium on. pp. 130--139. IEEE (2012)

\bibitem{cameron2014we}
Cameron, C., Singer, J.: We are all economists now: economic utility for
  multiple heap sizing. In: Proceedings of the 9th International Workshop on
  Implementation, Compilation, Optimization of Object-Oriented Languages,
  Programs and Systems PLE. p.~3. ACM (2014)

\bibitem{chekuri2005polynomial}
Chekuri, C., Khanna, S.: A polynomial time approximation scheme for the
  multiple knapsack problem. SIAM Journal on Computing  \textbf{35}(3),
  713--728 (2005)

\bibitem{clarke1971multipart}
Clarke, E.H.: Multipart pricing of public goods. Public Choice  \textbf{11}(1),
   17--33 (1971)

\bibitem{CloudSigmaPricing}
CloudSigma: Cloudsigma cloud pricing. \url{https://www.cloudsigma.com/pricing/}
  (2018), accessed: 2018-07-25

\bibitem{cortez2017resource}
Cortez, E., Bonde, A., Muzio, A., Russinovich, M., Fontoura, M., Bianchini, R.:
  Resource central: Understanding and predicting workloads for improved
  resource management in large cloud platforms. In: Proceedings of the 26th
  Symposium on Operating Systems Principles. pp. 153--167. ACM (2017)

\bibitem{dolev2012no}
Dolev, D., Feitelson, D.G., Halpern, J.Y., Kupferman, R., Linial, N.: No
  justified complaints: On fair sharing of multiple resources. In: Proceedings
  of the 3rd Innovations in Theoretical Computer Science Conference. pp.
  68--75. ITCS '12, ACM (2012). \doi{10.1145/2090236.2090243},
  \url{http://dx.doi.org/10.1145/2090236.2090243}

\bibitem{fonseca2017faircloud}
Fonseca, A., Sim{\~a}o, J., Veiga, L.: Faircloud: truthful cloud scheduling
  with continuous and combinatorial auctions. In: OTM Confederated
  International Conferences" On the Move to Meaningful Internet Systems". pp.
  68--85. Springer (2017)

\bibitem{fukuta2011vcg}
Fukuta, N.: Toward a vcg-like approximate mechanism for large-scale multi-unit
  combinatorial auctions. In: Proceedings of the 2011 IEEE/WIC/ACM
  International Conferences on Web Intelligence and Intelligent Agent
  Technology - Volume 02. pp. 317--322. WI-IAT '11, IEEE Computer Society
  (2011). \doi{10.1109/wi-iat.2011.191},
  \url{http://dx.doi.org/10.1109/wi-iat.2011.191}

\bibitem{funaro2016ginseng}
Funaro, L., {Agmon Ben-Yehuda}, O., Schuster, A.: Ginseng: market-driven {LLC}
  allocation. In: Proceedings of the 2016 USENIX Conference on Usenix Annual
  Technical Conference. pp. 295--308. USENIX Association (2016)

\bibitem{funaro2019stochastic}
Funaro, L., {Agmon Ben-Yehuda}, O., Schuster, A.: Stochastic resource
  allocation. In: Proceedings of the 15th ACM SIGPLAN/SIGOPS International
  Conference on Virtual Execution Environments (VEE '19). USENIX Association,
  ACM (2019)

\bibitem{gao2017iterative}
Gao, C., Lu, G., Yao, X., Li, J.: An iterative pseudo-gap enumeration approach
  for the multidimensional multiple-choice knapsack problem. European Journal
  of Operational Research  \textbf{260}(1),  1--11 (2017)

\bibitem{ghassemi2018exact}
Ghassemi-Tari, F., Hendizadeh, H., Hogg, G.L.: Exact solution algorithms for
  multi-dimensional multiple-choice knapsack problems. Current Journal of
  Applied Science and Technology  (2018)

\bibitem{ghodsi2011dominant}
Ghodsi, A., Zaharia, M., Hindman, B., Konwinski, A., Shenker, S., Stoica, I.:
  Dominant resource fairness: Fair allocation of multiple resource types. In:
  Nsdi. vol.~11, pp. 24--24 (2011)

\bibitem{gonen2000optimal}
Gonen, R., Lehmann, D.: Optimal solutions for multi-unit combinatorial
  auctions: Branch and bound heuristics. In: Proceedings of the 2Nd ACM
  Conference on Electronic Commerce. pp. 13--20. EC '00, ACM (2000).
  \doi{10.1145/352871.352873}, \url{http://dx.doi.org/10.1145/352871.352873}

\bibitem{GoogleComputePricing}
Google: Google cloud compute engine pricing.
  \url{https://cloud.google.com/compute/pricing} (2018), accessed: 2018-07-25

\bibitem{goudarzi2011multi}
Goudarzi, H., Pedram, M.: Multi-dimensional sla-based resource allocation for
  multi-tier cloud computing systems. In: Cloud Computing (CLOUD), 2011 IEEE
  International Conference on. pp. 324--331. IEEE (2011)

\bibitem{groves1973incentives}
Groves, T.: Incentives in teams. Econometrica: Journal of the Econometric
  Society pp. 617--631 (1973)

\bibitem{gutman2012fair}
Gutman, A., Nisan, N.: Fair allocation without trade. In: Proceedings of the
  11th International Conference on Autonomous Agents and Multiagent Systems -
  Volume 2. pp. 719--728. AAMAS '12, International Foundation for Autonomous
  Agents and Multiagent Systems (2012),
  \url{http://portal.acm.org/citation.cfm?id=2343799}

\bibitem{hifi2004heuristic}
Hifi, M., Michrafy, M., Sbihi, A.: Heuristic algorithms for the multiple-choice
  multidimensional knapsack problem. Journal of the Operational Research
  Society  \textbf{55}(12),  1323--1332 (2004)

\bibitem{hifi2004exact}
Hifi, M., Sadfi, S., Sbihi, A.: {An exact algorithm for the multiple-choice
  multidimensional knapsack problem}. Cahiers de la Maison des Sciences
  Economiques b04024, Université Panthéon-Sorbonne (Paris 1) (Mar 2004)

\bibitem{hines2011applications}
Hines, M.R., Gordon, A., Silva, M., Da~Silva, D., Ryu, K., Ben-Yehuda, M.:
  Applications know best: Performance-driven memory overcommit with {Ginkgo}.
  In: 2011 IEEE 3rd International Conference on Cloud Computing Technology and
  Science (CloudCom). pp. 130--137. IEEE (2011).
  \doi{10.1109/cloudcom.2011.27},
  \url{http://dx.doi.org/10.1109/cloudcom.2011.27}

\bibitem{jia2009revenue}
Jia, J., Zhang, Q., Zhang, Q., Liu, M.: Revenue generation for truthful
  spectrum auction in dynamic spectrum access. In: Proceedings of the tenth ACM
  international symposium on Mobile ad hoc networking and computing. pp. 3--12.
  ACM (2009)

\bibitem{kellerer2004introduction}
Kellerer, H., Pferschy, U., Pisinger, D.: Introduction to NP-Completeness of
  Knapsack Problems, pp. 483--493. Springer Berlin Heidelberg (2004)

\bibitem{khan2002solving}
Khan, S., Li, K.F., Manning, E.G., Akbar, M.M.: Solving the knapsack problem
  for adaptive multimedia systems. Stud. Inform. Univ.  \textbf{2}(1),
  157--178 (2002)

\bibitem{CloudSigmaPriceChart}
Kovacs, K.: Charting cloudsigma burst prices.
  \url{https://kkovacs.eu/cloudsigma-burst-price-chart} (2018), accessed:
  2018-07-25

\bibitem{lavi2003towards}
Lavi, R., Mu'Alem, A., Nisan, N.: Towards a characterization of truthful
  combinatorial auctions. In: Foundations of Computer Science, 2003.
  Proceedings. 44th Annual IEEE Symposium on. pp. 574--583. IEEE (2003)

\bibitem{lawler1979fast}
Lawler, E.L.: Fast approximation algorithms for knapsack problems. Mathematics
  of Operations Research  \textbf{4}(4),  339--356 (1979)

\bibitem{lazar1999design}
Lazar, A.A., Semret, N.: Design and analysis of the progressive second price
  auction for network bandwidth sharing. Telecommunication Systems---Special
  issue on Network Economics  (1999)

\bibitem{lee2007precise}
Lee, C.B., Snavely, A.E.: Precise and realistic utility functions for
  user-centric performance analysis of schedulers. In: Proceedings of the 16th
  International Symposium on High Performance Distributed Computing. pp.
  107--116. ACM (2007)

\bibitem{lueker1978data}
Lueker, G.S.: A data structure for orthogonal range queries. In: 19th Annual
  Symposium on Foundations of Computer Science, 1978. pp. 28--34. IEEE (1978)

\bibitem{maille2004multi}
Maill\'{e}, P., Tuffin, B.: Multi-bid auctions for bandwidth allocation in
  communication networks. In: IEEE INFOCOM (2004)

\bibitem{maille2007vcg}
Maille, P., Tuffin, B.: Why vcg auctions can hardly be applied to the pricing
  of inter-domain and ad hoc networks. In: 3rd EuroNGI Conference on Next
  Generation Internet Networks. pp. 36--39. IEEE (2007)

\bibitem{AzureBurstableVM}
Microsoft: Microsoft azure {AKS} b-series burstable {VM}.
  \url{https://azure.microsoft.com/en-us/blog/introducing-b-series-our-new-burstable-vm-size/}
  (2018), accessed: 2018-07-25

\bibitem{moser1997algorithm}
Moser, M., Jokanovic, D.P., Shiratori, N.: An algorithm for the
  multidimensional multiple-choice knapsack problem. IEICE Transactions on
  Fundamentals of Electronics, Communications and Computer Sciences
  \textbf{80}(3),  582--589 (1997)

\bibitem{mualem2008truthful}
Mu'alem, A., Nisan, N.: Truthful approximation mechanisms for restricted
  combinatorial auctions. Games and Economic Behavior  \textbf{64}(2),
  612--631 (2008). \doi{10.1016/j.geb.2007.12.009},
  \url{http://dx.doi.org/10.1016/j.geb.2007.12.009}

\bibitem{nisan2007computationally}
Nisan, N., Ronen, A.: Computationally feasible vcg mechanisms. Journal of
  Artificial Intelligence Research  \textbf{29},  19--47 (2007)

\bibitem{packetPacketSpot}
Packet: Packet cloud spot instances.
  \url{https://www.packet.net/bare-metal/deploy/spot/} (2018), accessed:
  2018-06-02

\bibitem{RackspaceCloudFlavors}
Rackspace: Rackspace cloud flavors.
  \url{https://developer.rackspace.com/docs/cloud-servers/v2/general-api-info/flavors/}
  (2018), accessed: 2018-09-27

\bibitem{razzazi2008exact}
Razzazi, M.R., Ghasemi, T.: An exact algorithm for the multiple-choice
  multidimensional knapsack based on the core. In: Advances in Computer Science
  and Engineering. pp. 275--282. Springer (2008)

\bibitem{roberts1979characterization}
Roberts, K.: The characterization of implementable choice rules. Aggregation
  and revelation of preferences  \textbf{12}(2),  321--348 (1979)

\bibitem{sanghavi2004optimal}
Sanghavi, S., Hajek, B.: Optimal allocation of a divisible good to strategic
  buyers. In: 43rd IEEE Conference on Decision and Control-CDC (2004)

\bibitem{sbihi2007best}
Sbihi, A.: A best first search exact algorithm for the multiple-choice
  multidimensional knapsack problem. Journal of Combinatorial Optimization
  \textbf{13}(4),  337--351 (2007)

\bibitem{skowron2013non}
Skowron, P., Rzadca, K.: Non-monetary fair scheduling: a cooperative game
  theory approach. In: Proceedings of the twenty-fifth annual ACM symposium on
  Parallelism in algorithms and architectures. pp. 288--297. ACM (2013)

\bibitem{souma2001universal}
Souma, W.: Universal structure of the personal income distribution. Fractals
  \textbf{9}(04),  463--470 (2001)

\bibitem{vickrey1961counterspeculation}
Vickrey, W.: Counterspeculation, auctions, and competitive sealed tenders. The
  Journal of Finance  \textbf{16}(1),  8--37 (1961)

\bibitem{wilkes2009utility}
Wilkes, J.: Utility Functions, Prices, and Negotiation. New York: Wiley (2009)

\bibitem{xiao2013dynamic}
Xiao, Z., Song, W., Chen, Q., et~al.: Dynamic resource allocation using virtual
  machines for cloud computing environment. IEEE Trans. Parallel Distrib. Syst.
   \textbf{24}(6),  1107--1117 (2013)

\bibitem{yang2007vcg}
Yang, S., Hajek, B.: Vcg-kelly mechanisms for allocation of divisible goods:
  Adapting vcg mechanisms to one-dimensional signals. IEEE Journal on Selected
  Areas in Communications  \textbf{25}(6) (2007)

\bibitem{yang2006cramm}
Yang, T., Berger, E.D., Kaplan, S.F., Eliot: {CRAMM}: Virtual memory support
  for garbage-collected applications. In: Proceedings of the 7th Symposium on
  Operating Systems Design and Implementation. pp. 103--116. OSDI '06, USENIX
  Association (2006)

\bibitem{ye2015rochester}
Ye, C., Brock, J., Ding, C., Jin, H.: Rochester elastic cache utility ({RECU}):
  Unequal cache sharing is good economics. International Journal of Parallel
  Programming pp. 1--15 (2015)

\bibitem{zhang2014dynamic}
Zhang, L., Li, Z., Wu, C.: Dynamic resource provisioning in cloud computing: A
  randomized auction approach. In: IEEE Infocom Proceedings. IEEE Computer
  Society (2014)

\bibitem{zhu2006utility}
Zhu, X., Wang, Z., Singhal, S.: Utility-driven workload management using nested
  control design. In: American Control Conference, 2006. IEEE (2006)

\end{thebibliography}

\end{document}